\documentclass[10pt,journal,compsoc]{IEEEtran}
\usepackage{xcolor}
\usepackage{threeparttable}
\usepackage{amsthm}
\usepackage{thmtools, thm-restate}

\declaretheorem{theorem}
\declaretheorem{lemma}

\declaretheorem{claim}

\ifCLASSOPTIONcompsoc
  \usepackage[nocompress]{cite}
\else
  \usepackage{cite}
\fi
\ifCLASSINFOpdf
  \usepackage[pdftex]{graphicx}
\else
  \usepackage[dvips]{graphicx}
\fi
\usepackage{amsmath}
\usepackage{algorithm}
\usepackage{algorithmic}

\ifCLASSOPTIONcompsoc
  \usepackage[caption=false,font=footnotesize,labelfont=sf,textfont=sf]{subfig}
\else
  \usepackage[caption=false,font=footnotesize]{subfig}
\fi
\usepackage{url}
\begin{document}
\title{An Efficient and Robust Committee Structure \\ for Sharding Blockchain}

\author{Mengqian~Zhang,~\IEEEmembership{IEEE~Member,}
        Jichen~Li,~\IEEEmembership{IEEE~Member,}
        Zhaohua~Chen,~\IEEEmembership{IEEE~Member,}
        Hongyin~Chen,~\IEEEmembership{IEEE~Member,}
        and~Xiaotie~Deng,~\IEEEmembership{IEEE~Fellow}% <-this % stops a space
\IEEEcompsocitemizethanks{\IEEEcompsocthanksitem M. Zhang is in the Department of Computer Science and Engineering, Shanghai Jiao Tong University, Shanghai, 200240, China (E-mail: mengqian@sjtu.edu.cn).\protect\\
%\IEEEcompsocthanksitem M. Shell was with the Department of Electrical and Computer Engineering, Georgia Institute of Technology, Atlanta, GA, 30332.\protect\\
%E-mail: see http://www.michaelshell.org/contact.html
\IEEEcompsocthanksitem J. Li, Z. Chen, H. Chen and X. Deng are in the Center on Frontiers of Computing Studies, Computer Science Dept., Peking University, Beijing, 100871, China (E-mail: \{2001111325, chenzhaohua, chenhongyin, xiaotie\}@pku.edu.cn).\protect\\

\IEEEcompsocthanksitem X. Deng is the corresponding author. M. Zhang, J. Li, Z. Chen and H. Chen contributed equally to this work.
%\IEEEcompsocthanksitem J. Doe and J. Doe are with Anonymous University.
}% <-this % stops an unwanted space

\thanks{This work is supported by Science and Technology Innovation 2030 - “New Generation Artificial Intelligence” Major Project No.(2018AAA0100901).}}

\IEEEtitleabstractindextext{%
\begin{abstract}
%Traditional blockchain systems suffer from serious scalability problems. There is a huge gap between their transaction throughput and practical needs. In recent years, sharding is deemed as a promising way to overcome this problem. By partitioning all nodes into small committees and letting them work in parallel, sharding can help to significantly lower the amount of computation, reduce the overhead on each node’s storage, as well as enhance the throughput of blockchain.
Nowadays, sharding is deemed as a promising way to save traditional blockchain protocols from their low scalability. However, such technique also brings several potential risks and huge communication overheads. An improper design may give rise to the inconsistent state among different committees. Further, the communication overheads arising from cross-shard transactions unfortunately reduce the system’s performance. In this paper, we first summarize five essential issues that all sharding blockchain designers face. For each issue, we discuss its key challenge and propose our suggested solutions. In order to break the performance bottlenecks, we propose a reputation mechanism for selecting leaders. The term of reputation in our design reflects each node's honest computation resources. In addition, we introduce a referee committee and partial sets in each committee, and design a recovery procedure in case the leader is malicious. Under the design, we prove that malicious leaders will not hurt the system and will be evicted. Furthermore, we conduct a series of simulations to evaluate our design. The results show that selecting leaders by the reputation can dramatically improve the system performance.
\end{abstract}

}

\maketitle

\IEEEdisplaynontitleabstractindextext

% \fi
\IEEEpeerreviewmaketitle

\IEEEraisesectionheading{\section{Introduction}\label{sec:introduction}}

\IEEEPARstart{T}{he} blockchain revolution has established a milestone with Nakamoto's Bitcoin~\cite{nakamoto2019bitcoin} during the long search for a dreaming digital currency. It maintains the distributed database in a decentralized manner and has achieved a certain maturity to provide the Internet community with a service that captures the most important features of cash: secure, anonymity, easy to carry, easy to change hand, and exchangeable across  national boundaries. At the same time, Bitcoin realizes the tamper proof property needed for transactions. 

%Low scalability is a fatal flaw for today's most popular blockchains. In the context of blockchains, a system is considered scalable if its throughput (transactions per second or \textit{TPS} for short) grows with the increase of computing resources. In most blockchain projects, all nodes in the network have to agree on a certain set of transactions. Therefore, only a rather fixed amount of transactions can be included in a block over a period of time, regardless of the number of nodes in the network. Such design leads to a quite low throughput, rendering major blockchain systems (including Bitcoin~\cite{nakamoto2019bitcoin} and Ethereum~\cite{wood2014ethereum}) suffer from serious scalability problem. 

Unfortunately, state-of-the-art blockchain systems (\textit{e.g.}, Bitcoin~\cite{nakamoto2019bitcoin} and Ethereum~\cite{wood2014ethereum}) suffer from serious scalability problem. In the context of blockchains, a system is considered scalable if its throughput (transactions per second or \textit{tps} for short) grows with the increase of computing resources. In most blockchain projects, all nodes in the network have to agree on a certain set of transactions. Therefore, only a rather fixed amount of transactions can be included in a block over a period of time, regardless of the number of nodes in the network.

%Unfortunately, the transaction throughput of Bitcoin is 3 to 4 orders of magnitude lower than those for centralized payment systems such as Visa. This weakness limits its popularization. Specifically, Bitcoin can process only 3.3-7 transactions per second in 2016~\cite{CromanDEGJKMSSS16} while Visa handles more than 24,000 transactions per second~\cite{Visa}. The low scalability of Bitcoin is the key factor causing this displeasing result.

The classic design leads to a quite low throughput. The tps of Bitcoin is 3 to 4 orders of magnitude lower than those for centralized payment systems such as Visa. Specifically, Bitcoin can process only 3.3-7 transactions per second in 2016~\cite{CromanDEGJKMSSS16} while Visa handles more than 24,000 transactions per second~\cite{Visa}. Such weakness severely limits the popularization of blockchain.

To overcome this problem, a promising way is to borrow the idea of \textit{sharding} from the field of database. Sharding is a well-known technique to build the scale-out database. The main idea is to break up large tables into smaller chunks and spread them across multiple servers. It enables the sharing of overall workload and greatly increases the performance. 

Inspired by this idea, researchers in the area of blockchain apply sharding to improve blockchains' scalability in recent years. Similarly, we can partition nodes into parallel committees and have each committee maintain a subset of transactions. Under this method, the transaction processing rate is proportional to the number of committees rather than a constant. Therefore, less computational power is wasted, higher processing capacity is achieved and the system scales well. %Currently, sharding is considered as an optimistic way to help blockchain scale well in literature.

%Sharding is well known in building the “scale-out” database by separating the whole state into multiple parts and making them work concurrently. Inspired by this idea, the use of sharding is applied to improve distributed ledger's performance in these years. When nodes are partitioned into groups, less computational power is wasted and higher processing capacity is achieved. 

%However, most sharding-based protocols fail to maintain high efficiency, especially in the presence of betrays by important committee leaders in charge of blockchain security~\cite{Kokoris-KogiasJ18,ZamaniM018}. They also fail to give an explicit incentive for nodes to join the protocol and behave honestly.

%We present CycLedger, a fully decentralized payment-processor that provides the scalability over the number of participating nodes and a proof of security. Specifically, our protocol  remains robust when leaders in every committee are faulty. This is a problem that remains unsolved until now. Meanwhile, we hope there is enough encouragement for a node to participate honestly in our protocol. Furthermore, scalability is realized via sharding nodes into concurrent committees as previous works have shown~\cite{Kokoris-KogiasJ18,LuuNZBGS16,ZamaniM018}. 

Despite its benefits, the application of sharding in blockchain also imposes a great deal of complexity and potential risks. Rather than accessing and updating data from local storage, nodes have to process transactions across multiple shards. An improper design may cause the inconsistency of data, which is an enormous threat to the system. Thus, it is of great importance and urgency to summarize and present general design process, principles and suggestions for the sharding protocols. In addition, %existing shading blockchains mainly take the classic Byzantine Fault Tolerance (BFT) algorithm for consensus. 
cross-shard transactions make it inevitable to carry out large amount of communications among different committees, which unfortunately reduces the performance of the system. This is a crucial issue that has not been properly resolved until now.
 In this paper, we identify five essential issues in sharding blockchain, which together composite a complete protocol. For each issue, we discuss its key challenge, give a brief review on existing approaches, and propose our suggested solution.
 By doing numerical studies, we show that cross-shard transactions take an overwhelming fraction in the sharding blockchain. These transactions inevitably bring heavy overheads on communicating, which is the task of leader in each committee. As a result, their capability becomes a major bottleneck of the whole protocol. 
 
 For better performance, we design a reputation mechanism. Specifically, the reputation of a node reflects its honest computation resource. In each round, those nodes with higher reputation will be selected as leaders. For safety and liveness of the protocol, we introduce a unique referee committee and partial sets in each committee. Briefly, the referee committee act as a arbitrator between opposing parties, and partial sets take the responsibility to supervise the leader's behavior. Further, we propose a recovery procedure, which is triggered upon a malicious leader is detected. As a consequence, the malicious leader is evicted and a new leader is re-selected immediately, ensuring the system running properly.
 
 Theoretically, we prove that selecting leaders via reputation and the committee assignment scheme are secure with overwhelming probability. Also, under such design, the transaction processing procedure satisfies safety and liveness. Further simulation results support that the reputation mechanism could effectively help to filter out those nodes with abundant computational resources. By arranging these nodes as leaders, each committee could process more transactions and further improve the system performance.  %analytically and experimentally

This paper is organized as follows. In Section~\ref{sec:issues}, we discuss five key issues in sharding blockchain. Section~\ref{sec:model} states system model and elaborate the problem we aim to solve. We propose our solution in Section~\ref{sec:solution}, including the reputation mechanism and the recovery procedure. After that, we give the theoretical analysis on our design in Section~\ref{sec:analysis} and conduct a series
of simulations in Section~\ref{sec:simulation}. Finally, we review previous studies in Section~\ref{sec:related} and conclude the paper in Section~\ref{sec:conclusion}. 

\section{Issues in Sharding Blockchains}\label{sec:issues}

When it comes to sharding blockchain protocols, there are five basic issues to consider. Like previous works, we will analyze these problems based on the \textit{UTXO} (Unspent Transaction Output) model. A UTXO indicates the amount of digital currency that its owner can use later. Specifically, a transaction will take one or multiple UTXOs as inputs, and then output new UTXOs. And each UTXO can only be spent once.
%is the amount of digital currency remaining after a transaction is executed. Under such model, any transaction takes one or multiple UTXOs as input (resembling the spent of assets for one or multiple parties), and outputs another one or multiple UTXOs (resembling the gain of assets for another one or multiple parties). 

\subsection{Issue \#1: What to split?}\label{sec:issue1}
%The first thing is to think what should be partitioned. In line with the original intention of the sharding technology, the following is our answer:
%\begin{itemize}
%	\item \textbf{Nodes}. Sharding-based blockchains divide nodes into disjoint groups, each of which are called a \textit{committe}. As a result, all committees work in parallel, to increase the processing speed.
%	\item \textbf{Pending transactions}. Instead of storing all transactions into mempool by every node, the pending transactions are also partitioned into different groups, which we call \textit{shards}. It is worth noting that the number of shards is equal to the committees'. Each shard of transactions is just stored and processed by the corresponding committee. 
%	\item \textbf{Blockchain history}. To lessen the storage pressure of nodes, the heavy blockchain data are also split and each committee maintains a sub-chain.
%\end{itemize}

The first issue is to determine what to partition. Before getting into the details, we first distinguish a pair of essential concepts. In the context, when referring to the term \textit{nodes}, we mean the participants who invest resources (hardware, electricity, etc) and offer services (such as processing transactions) in blockchain systems, similar to the miner in Bitcoin. On the other hand, the participants who use services (\textit{e.g.}, generating transactions) are known as \textit{users}. %resembling a transaction party in Bitcoin. 

In line with the original intention of the sharding technology, we answer the question as follows:
\begin{itemize}
	\item \textbf{Nodes}. A sharding blockchain divides nodes into disjoint groups, each of which is known as a \textit{committee}. All committees work concurrently, therefore increasing the transaction processing speed.
	\item \textbf{Pending transactions}. Instead of being kept in the mempool %(the collaborative waiting area for all pending transactions) 
	of every node, pending transactions are also partitioned into different groups, which we call \textit{shards}. It is worth noting that there is a one-to-one mapping from shards to committees. Each shard of transactions is only stored and processed by the nodes in the corresponding committee. 
	\item \textbf{Blockchain history}. To alleviate the full storage burden on nodes, the heavy blockchain data is also split and maintained by separated committees, in the form of a sub-blockchain with lower volume.
\end{itemize}

\subsection{Issue \#2: How to split?}\label{sec:issue2}

As a subsequence, it is important to ponder the methods to partition the nodes, pending transactions, and blockchain history. 

Under the existence of potential Byzantine attackers, nodes are expected to be distributed uniformly into committees. The uniform assignment keeps all committees under a safe Byzantine fraction, and drives away the possibility for malicious nodes to gather in and take over certain committee. To realize a uniformly random assignment, the key is to generate a trusted randomness distributedly, which has long been a critical task in blockchain. Several protocols~\cite{BentovPS16a,DavidGKR18} use an external cryptographic hash function, which takes an unpredictable and tamper-resistant value (\textit{e.g.}, the Merkle root of transactions in the previous block) as input, and the output of the function is the randomness. Other attempts include Verifiable Random Functions (VRFs)~\cite{MicaliRV99} in Algorand~\cite{GiladHMVZ17}, and Publicly Verifiable Secret Sharing (PVSS) schemes~\cite{Schoenmakers99,CascudoD17} in Ouroboros~\cite{KiayiasRDO17} and OmniLedger~\cite{Kokoris-KogiasJ18}. 

For the purposes of load balancing, the transaction is assigned to one committee based on its id, \textit{i.e.}, the hash value of the transaction. Specific designs vary among different protocols. One method is to use the first $k$ bits of the transaction id to determine its shard, assuming there are $2^k$ shards (and committees) in the system. For example, all transactions beginning with 010 in their hash belong to shard \#2 (here $k = 3$). We can also determine the belonging shard of a transaction by calculating its id modulo $2^k$. To prevent unbalanced loading caused by deliberate users when generating transactions, some protocols add salt when calculating the hash of transactions~\cite{ManuskinME20}.

We emphasize that, under any proper sharding configuration, each transaction should be precisely maintained by a unique committee, to avoid consistency issues\footnote{Or else, different committees may come up with conflicting results concerning the same transaction.}. Accordingly, the blockchain history is also separated, and each virtual UTXO is exactly kept by a single committee.

\subsection{Issue \#3: How to deal with the dynamic membership?}\label{sec:issue3}

As discussed previously, the sharding technique aims at solving the scalability issue for permissionless blockchains, in which nodes can join and leave freely. In this scenario, the attackers may execute sufficiently many join and leave operations to manipulate a certain committee. Thus, it is necessary to adjust the committee assignments periodically. A simple idea is to reshuffle all nodes at the beginning of each round. Such idea has several shortcomings. First, a full bootstrapping procedure will stop the whole system. All nodes halt until new committees are formed. Second, frequent reorganization brings expensive communication overheads. Specifically, nodes have to change connection with different sets of committee members. Also, they need to fetch data of the new ledger from other nodes.

To balance usability and performance, replacing only a subset of committee members is a better alternative. Such solution deals with nodes' joining and leaving smoothly, and resists against the slowly-adaptive Byzantine adversary. As a practice, RapidChain~\cite{ZamaniM018} borrows the idea of Cuckoo rule~\cite{AwerbuchS09}, requiring that only a constant number of nodes is switched to other committees in each round. Yet, this method leaves the adversary enough room to take over a certain committee and further cause damage to the whole system. More specifically, at some round $r$, the adversary may plan to corrupt all members in a committee during round $r+d$. Afterwards, even if a constant number of members are switched out, the majority of this committee will be malicious.

%Here, we propose the ECFR (Expected Constant-Fraction Reshuffling) scheme with parameter $\alpha\in (0, 1)$. In expectation, ECFR requires a constant fraction of members inside a committee to be switched out in each round. Precisely:

In order to solve the dynamic membership issue properly, we propose an Expected Constant-Fraction Reshuffling (ECFR) scheme with parameter $\alpha$. In expectation, it requires a constant fraction of members inside a committee to be switched out in each round. Precisely, the $\mathbf{\alpha}$\textbf{-ECFR scheme} is designed as follows: Let $\alpha \in (0, 1)$ be a constant. In each round, we independently mark each node with probability $\alpha$, and uniformly reassign all \textit{marked} nodes to all committees, assuring that all committees have the same size.

We mention that the number of reshuffled nodes of $\alpha$-ECFR scheme lies between total reshuffling and RapidChain's solution. With little bias of notation, total reshuffling is an extreme case of $\alpha$-ECFR scheme with $\alpha = 1$. Further, there exists a constant $\alpha$ enabling that with high probability honest nodes take the majority in all committees. We prove this security property in Section~\ref{sec:reshuffling}.

%Further, we present the following theorem to demonstrate the effectiveness of ECFR scheme. Here, the mentioned \textit{honesty} means the property that honest nodes take the majority in all committees.
%for the simplicity of expression, we define the \textit{honesty} property for each round of a sharding protocol.
%\begin{itemize}
%    \item \textit{Honesty.} Honest nodes take the majority in all committees.
%\end{itemize}

%\begin{restatable}{theorem}{ECFR}\label{the:ECFR}
%Under our system model (see Section~\ref{sec:model}), there exists a constant $\alpha\in (0, 1)$, such that under $\alpha$-ECFR scheme, any round satisfies honesty with high probability, given that all previous rounds satisfy honesty.
%\end{restatable}

%\begin{theorem}\label{the:ECFR}
%Under our system model (see Section~\ref{sec:model}), there exists a constant $\alpha$, such that under $\alpha$-ECFR scheme, any round satisfies honesty with high probability, given that all previous rounds satisfy honesty.
%\end{theorem}

%The proof of the theorem is deferred to Appendix~\ref{app:ECFR}.

\subsection{Issue \#4: How to process intra-shard transactions?}\label{sec:issue4}

Based on the \textit{UTXO} model, a transaction specifies one or more outputs of previous transactions as its input(s). For a transaction $\mathsf{tx}$ charged by some committee $C$, if all its referenced transactions are also maintained in $C$, then $\mathsf{tx}$ is known as an \textit{intra-shard transaction}. As a result, committee members can verify the intra-shard transactions by themselves and reach a consensus afterwards. There are a lot of options for the specific consensus algorithm. On one hand, each committee can be regarded as a smaller blockchain system with less nodes and transactions. Thus, the typical blockchain consensus algorithms can be applied here, such as Proof of Work (PoW), Proof of Stake (PoS) and so on. The discussion on these consensus protocols is out of the scope of this paper. More details can be found in~\cite{abs-2001-07091}.

On the other hand, several protocols which are not suitable for the large-scale blockchain systems find their place in this scenario. A typical example is the classic Byzantine Fault Tolerance (BFT) consensus protocol. As the first practical solution, \cite{CastroL99} correctly survives Byzantine faults in asynchronous networks, allowing a group of nodes to reach agreement on some value. However, it only works in the setting where all participants are known to each other, and the network size is small. Specifically, Such solution does not scale well because of its communication overhead. The protocol creates $O(n^2)$ communication complexity, where $n$ is the number of nodes, which becomes unacceptable when the size of the network increases to, for example, 10,000. Fortunately, these problems are mitigated when applied to each committee in a sharding design. Here are two main reasons, (1) It is easier for nodes to know each other in a smaller committee, and (2) the size of a committee is small enough to enable the BFT protocols to work well. As a result, the BFT protocols have been widely used in the sharding blockchain systems to reach intra-committee consensus.

\subsection{Issue \#5: How to process cross-shard transactions?}\label{sec:issue5}

In comparison to an intra-shard transaction, for a transaction $\mathsf{tx}$ charged by committee $C$, if $\mathsf{tx}$ has a referenced input transaction which is not maintained by $C$, then $\mathsf{tx}$ is called a \textit{cross-shard transaction}. Facing this kind of transactions, members of $C$ have to ask other committees for the validity of the inputs.

Existing studies give some solutions to this issue. In OmniLedger~\cite{Kokoris-KogiasJ18}, before submitting a cross-shard transaction to its assigned committee $C$, the user who generates it needs to collect \textit{proof-of-acceptance} from all referenced committees, whose members process the transaction and update the \textit{UTXOs} state. RapidChain~\cite{ZamaniM018} transforms the cross-shard transaction into intra-shard transactions. For each input $i$ outside $C$, a bridge transaction $tx'$ is created, which takes $i$ as input, and generates an output with the same value as $i$. However, $tx'$ is always assigned to $C$. Consequently, the leader of $C$ replaces the original cross-shard transaction by an intra-shard transaction which takes the outputs of all bridge transactions as inputs. CycLedger~\cite{ZhangLCCD20} transforms the cross-committee consensus into intra-committee consensus by having the leader of $C$ discussing on the relevant inputs with its belonged committee, and comes up with a consistent result for all relevant committees. Such solution resembles a two-phase commit (2PC), but in a decentralized manner.

%\begin{figure}[!t]
%    \centering
%    \includegraphics[width=2.5in]{Figure/arg_inputs_num2.png}
%    \caption{Average number of inputs per transaction over past ten years.}
%    \label{fig:arg_inputs_num}
%\end{figure}

\begin{figure}[!t]
    \centering
    \includegraphics[width=3.4in]{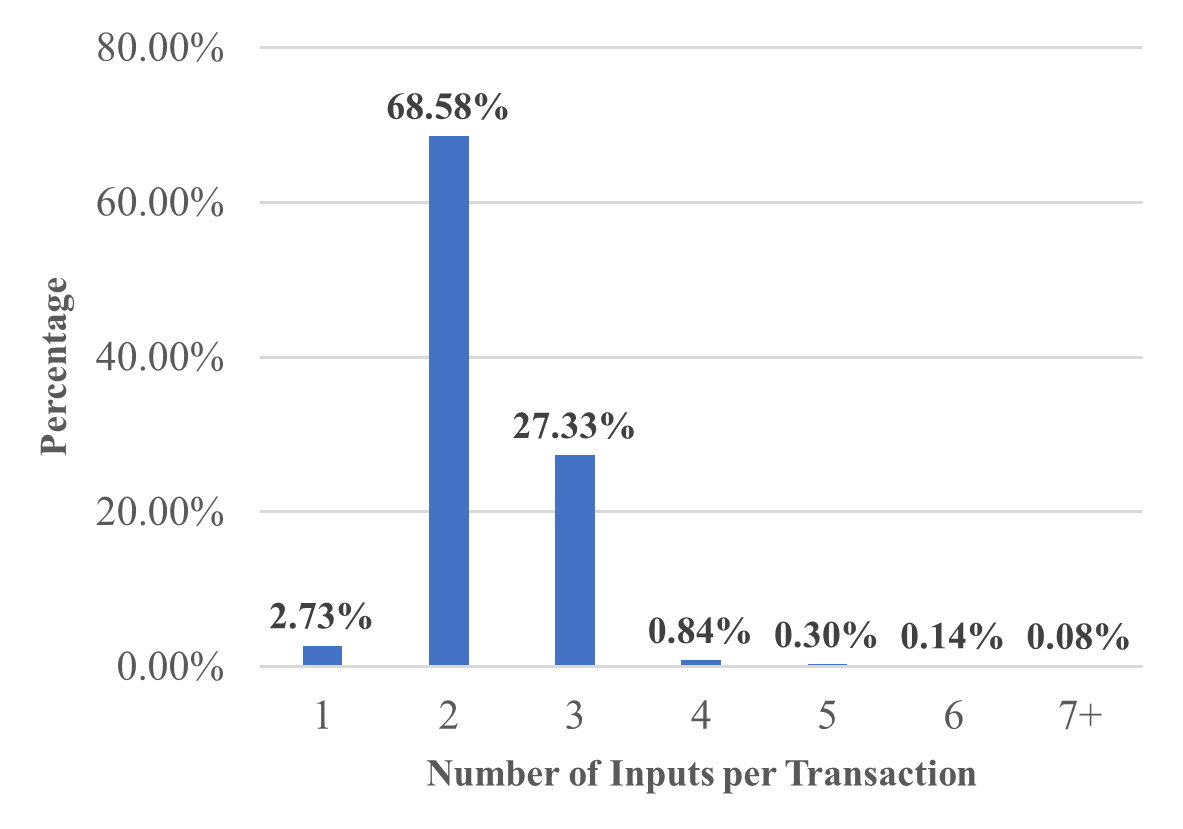}
    \caption{Percentage of the number of transactions' inputs in Bitcoin in the past decade.}
    \label{fig:inputs_num_cdf}
\end{figure}

%\begin{figure}[!t]
%    \centering
%    \includegraphics[width=2.7in]{Figure/per_cross_shard_txs.png}
%    \caption{Percentage of cross-shard transactions when the number of shards is between 10 to 100.}
%    \label{fig:percentage_cross_txs}
%\end{figure}

At last, we remark that cross-shard transactions bring extra communication overheads, which is one of the main disadvantages of sharding blockchains. We conduct several numerical studies to support this idea. According to the data from \textit{Bitcoin Visuals}\footnote{\url{https://bitcoinvisuals.com/}}, which runs a fully-validating bitcoin node and provides extensive statistics over the network, we visualize the information of transaction inputs over the past ten years. \figurename~\ref{fig:inputs_num_cdf} shows the percentages regarding the number of transactions' inputs in Bitcoin in the past decade. As can be seen, most transactions have two or three inputs. Furthermore, we download the full information of transactions from \textit{Blockchair}\footnote{\url{https://blockchair.com/}}. As \figurename~\ref{fig:frac_cross_txs} shows, when splitting transactions randomly into 20 shards, approximately 96\% of them are cross-shard transactions. When we further increase the shard number to 100, such fraction ascends correspondingly to above 99\%.

\section{System Model}\label{sec:model}

\subsection{Notation}
The blockchain system works in rounds. In each round $r$, suppose there are $n$ nodes in the network\footnote{In practice, $n$ is changing with $r$. However, we use $n$ instead of $n^r$, for simplicity.} and each of them has a reputation $w^r_1, w^r_2, \cdots, w^r_n$ which develops as the protocol executes. A unique \textit{referee committee} $C_R^r$ is selected in round $r - 1$ to manage nodes' identities and propose round $r$'s block $B^r$. At the same time, all other nodes are partitioned into $m$ \textit{common committees} which we denote as $C_1^r, C_2^r, \cdots, C_m^r$. In our protocol, all committees have the same size $c = \Theta(\log^2 n)$. Therefore we have $n = (m + 1)\cdot c$. Each common committee $C_i^r (1 \leq i \leq m)$ includes a \textit{leader} $l_i^r$, $\lambda$ partial set members and $c - \lambda - 1$ \textit{common members}. Partial set members in $C_i^r (1 \leq i \leq m)$ form the partial set of the committee, which is denoted as $C_{i, par}^r = \{c_{i, 1}^r, c_{i, 2}^r, \cdots, c_{i, \lambda}^r\}$. Readers can refer to \figurename \ref{fig:hie_stru} for the hierarchical configuration of the network. In the following sections, we omit the superscript of $r$ in notation if there is no ambiguity. %It is ensured that there is at least one non-faulty node in each partial set. In practice, $\lambda$ can be an appropriate value, like 40. 
 
%Besides, with overwhelming probability, each round is successfully terminated (\textit{i.e.}, all expected operations are finished) within a fixed time $T$.

\begin{figure}[!t]
    \centering
    \includegraphics[width=3.4in]{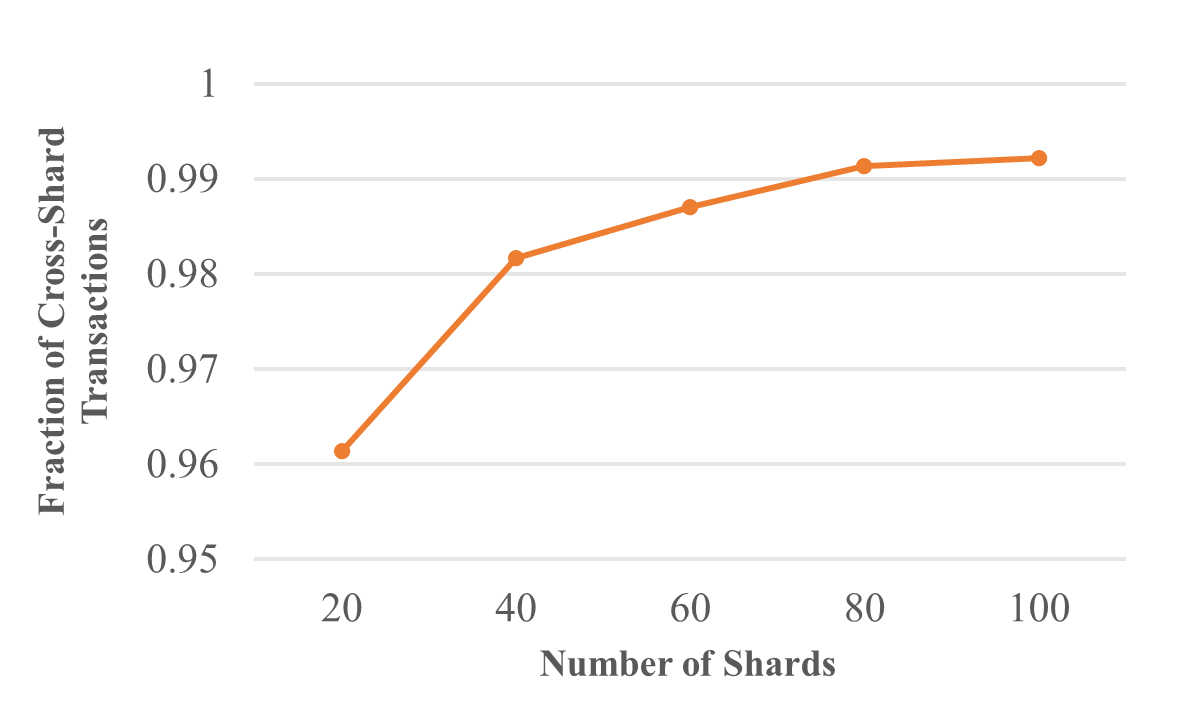}
    \caption{Fraction of cross-shard transactions when the number of shards is between 20 to 100 for transactions in Bitcoin.}
    \label{fig:frac_cross_txs}
\end{figure}

\subsection{Network Model}

%We assume the good connection within a committee. Meanwhile, all leaders and partial set members are linked. Furthermore, we suppose that each leader or partial set member is connected with the whole referee committee $C_R$. This requires a far less amount of reliable connection channels than other works~\cite{abs-1901-05741,Kokoris-KogiasJ18,LuuNZBGS16,ZamaniM018} in which they require a good connection among all honest nodes. Meanwhile, like existing works~\cite{Kokoris-KogiasJ18,ZamaniM018}, we assume synchronous communication within committees  (\textit{i.e.}, the delay of transporting every message is within some $\Delta$) which is realistic in real-world as a committee only consists of several hundred nodes. Meanwhile, all leaders and partial set members are synchronously linked, however, with a larger time delay of $\Gamma$. Concerning other connections, we only need to assume partially-synchronous channels~\cite{cachin2011introduction,CastroL99,Kokoris-KogiasJ18}.

We assume the well connection within a committee. Meanwhile, all leaders and partial set members are linked. Furthermore, we suppose that each leader or partial set member is connected with the whole referee committee $C_R$. Such assumption requires a far less amount of reliable connection channels than other works~\cite{abs-1901-05741,Kokoris-KogiasJ18,LuuNZBGS16,ZamaniM018} in which they require a reliable connection among all honest nodes. We assume synchronous communication within committees  which is realistic as a committee only consists of several hundred nodes. Meanwhile, all leaders and partial set members are synchronously linked. Concerning other connections, we only need to assume partially-synchronous channels~\cite{cachin2011introduction,CastroL99,Kokoris-KogiasJ18}.

\subsection{Threat Model}

We use a Public-Key Infrastructure (PKI) to give each node a public/secret key pair ($PK, SK$), which enables the digital signature scheme in communication. We assume the existence of a probabilistic polynomial-time \textit{Adversary} which controls a $f < 1/3$ fraction of total nodes. \textit{Corrupted} nodes may collude and act out arbitrary behaviors. The adversary can change the order of messages sent by non-faulty nodes under the restriction given in our network model. Other nodes, known as \textit{honest} nodes, always follow the protocol and do nothing outside the regulation. At the same time, we suppose the adversary to be mildly-adaptive, such that after appointing a corruption set, the set of nodes become malicious after $d$ rounds, where $d \geq 1$. Also, we assume all nodes in the network have access to an external random oracle $H$ which is collision-resistant.
%as well as a Verifiable Random Function (VRF) scheme~\cite{MicaliRV99}.

\subsection{Problem Definition}\label{sec:mod_prob}

%We assume that a large set of transactions are continuously sent to our network by external \textit{users}. Users are almost equally divided into $m$ shards. The status of each shard, including the users' identity and Unspent Transaction Outputs (UTXOs), is maintained by the corresponding committee. All processors have access to an authentication function $V$ to verify whether a transaction is legitimate, \textit{e.g.}, the sum of all inputs of the transaction is no less than the sum of all outputs and there is no double-spending. 

We assume that transactions are continuously sent to our network by the \textit{users}. All transactions and Unspent Transaction Outputs (UTXOs), are divided to $m$ shards, and maintained by the corresponding common committee. Besides, all nodes can verify whether a transaction is valid.

As discussed in section~\ref{sec:issue4}, most sharding blockchains adopt the classic BFT algorithm to reach consensus in a committee, in which the leader is a key role. Cross-shard transactions make it inevitable to carry out large amount of communications among different committees. This, unfortunately, puts huge strain on the leader of each committee and makes the leader a bottleneck of the system efficiency. In this work, we focus on this problem, aiming to improve the performance of sharding blockchain protocols.
\section{Solution}\label{sec:solution}

In this section, we propose our solution to the problem we discussed in Section~\ref{sec:mod_prob}. Specifically, in Section~\ref{sec:sol_rep}, we introduce a reputation mechanism to select leaders and provide explicit incentive for nodes in the system. In Section~\ref{sec:sol_rec}, we present a recovery procedure to ensure the safety and liveness properties. 
%Further, in Section~\ref{sec:sol_over}, we outline a sharding blockchain which apply the two schemes we designed.

%figure out those nodes with higher honest computational resources
% We assume the computing resources of nodes are different.

\subsection{Leader Selection via Reputation}\label{sec:sol_rep}
In a sharding blockchain, nodes in each committee are led by the leader and work together to process transactions. In our design, we introduce the concept of reputation which reflects the honest computational resources of each node and helps to select leaders of each committee. Here, we take the intra-shard transactions processing as an example to go into details about the design of node's reputation.

\subsubsection*{Phase 1: Voting and Decision Making}
%Consider the following scenario: In a committee of $c$ nodes, the leader $l$ collects a subset of intra-shard transactions $TXList$ and broadcasts it to everyone in the committee. 
 
%After receiving $TXList$, everyone gives its opinion on the validity of listed transactions. More specifically, the node votes \textit{Yes} for those transactions it agrees, \textit{No} for disagreed and \textit{Unknown} for the left. %When a node fails to judge a transaction within the given time, it should vote \textit{Unknown}. 
%Afterwards, each node forwards its voting list back to the leader. For the nodes who fail to reply within a certain time (\textit{e.g.} $6\Delta$), they are deemed as voting \textit{Unknown} on all transactions. This is to avoid malicious nodes from indefinitely delaying.
 
%With all voting lists, first the leader follows the majority rule and picks up the set of transactions with more than $\frac{c}{2}$ \textit{Yes}, which is named as $TXdecSET$. Then it integrates everyone's vote in a set named $VList$. At last the leader runs PBFT within this committee to reach consensus on both $TXdecSET$ and $VList$. 

Consider the following scenario: In a committee of $c$ nodes, the leader $l$ collects a subset of intra-shard transactions $TXList$ and broadcasts it to everyone in the committee. 
 
After receiving $TXList$, everyone gives its opinion on the validity of listed transactions. More specifically, the node votes \textit{Yes} for those transactions it agrees, \textit{No} for disagreed and \textit{Unknown} for the left\footnote{If a malicious node refuses to give any opinion, it is presumed to have voted \textit{Unknown}. Meanwhile, when an honest node fails to judge a transaction within the given time, it should vote \textit{Unknown}.}.
Afterwards, each node forwards its voting list back to the leader. For the nodes who fail to reply within a certain predetermined time, they are deemed as voting \textit{Unknown} on all transactions. This is to avoid malicious nodes from indefinitely delaying.
 
With all voting lists, first the leader follows the majority rule and picks up the set of transactions with more than $\frac{c}{2}$ \textit{Yes}, which is denoted as $TXdecSET$. Then it integrates everyone's vote in a set denoted as $VList$. At last the leader runs a BFT protocol (e.g., ~\cite{CastroL99}) within this committee to reach consensus on both $TXdecSET$ and $VList$. 

\subsubsection*{Phase 2: Reputation Updating}

As previously described, there are $c$ nodes in the committee and $VList$ contains $c$ votes. After each voting, the leader scores all members according to their votes and the final decision.

Suppose the amount of transactions to be determined is $D$, and let $+1$, $-1$ and $0$ represent \textit{Yes}, \textit{No} and \textit{Unknown} respectively. Subsequently, We can use a $D$-dimension vector to indicate a vote, with each entry representing the opinion on the corresponding transaction. Let $\mathbf{v}_i=\{v_{i,k}|k=1,2, \cdots, D\}$ denote the vote of member $i$, where $v_{i,k}$ is the member $i$'s opinion on the $k^{th}$ transaction. One's score is closely related to its voting accuracy. %We use the cosine value of the angle between two vectors to measure the proximity of two corresponding votes. 
Here, we define one's accuracy as the cosine similarity between its voting vector and the resulting vector, which is determined by the majority algorithm and denoted by $\mathbf{u}=\{u_{k}|k=1,2, \cdots, D\}$. Let $s_{i}$ be the member $i$'s score. We have

\begin{equation}
\begin{aligned}
s_{i} &= D \cdot \cos(\mathbf{v}_i,\mathbf{u}) + \mathsf{bonus}_i    \\
      &= D \cdot \frac{\sum_{k=1}^{D}v_{i,k} u_{k}}{\sqrt{\sum_{k=1}^{D}v_{i,k}^2}\sqrt{\sum_{k=1}^{D}u_{k}^2}} + \mathsf{bonus}_i.  
\end{aligned}
\label{eq:score}
\end{equation}

As \eqref{eq:score} shows, the score function contains two parts. The first part ensures that one's score is positively associated with the number of transactions processed in this voting and its accuracy. The second part is an extra bonus to award those nodes whose votes are totally the same with the final result, \textit{i.e.}, whose accuracy $\cos(\mathbf{v},\mathbf{u})=1$. Significantly, the bonus is a bit different for the leader and other members and it is designed as follows:
\begin{equation}
        \mathsf{bonus}_i = \left\{
            \begin{array}{rcl}
                \sigma \cdot D - \omega/D,       && {\cos(\mathbf{v}_i,\mathbf{u})=1,\ i \text{\ is\ a\ leader};}\\
                \sigma \cdot D, && {\cos(\mathbf{v}_i,\mathbf{u})=1,\ i\text{\ is\ a\ member};}\\
                0,       && {\cos(\mathbf{v}_i,\mathbf{u}) < 1.}\\
            \end{array} 
        \right.
        \label{eq:bonus}
\end{equation}
where $D$ is the number of transactions to be voted, $\sigma, \omega > 0$ are two predetermined parameters. There are two key points on the design of bonus function here:
\begin{itemize}
    \item A perfect non-leader member gains more bonus than the leader. It is based on the consideration that the member may own higher honest computation resources, yet covered up by the leader who sets an upper line on the number of processed transactions of the committee. A higher bonus creates the possibility for the member to stand out.
    \item The gain difference between a perfect non-leader member and the leader on reputation gets smaller with the rise on the number of processed transactions. As a higher number of processed transactions implies the better ability of a leader, it is reasonable to lessen the possibility that the well-performed leader gets replaced.
\end{itemize}

%Owing to the random generation of each committee and the consensus scheme inside a committee, the majority result reflects truthful members' views. The main idea here is, the closer one's opinion stands with the consensus, the higher the score it will get. Concretely, if a member has the same answers with the consentaneous results, it would be rewarded the highest score, \textit{i.e.}, $+1$. On the contrary, if a member replies with completely opposing opinions, it will face a loss of $-1$ in reputation.

After calculating all scores, the leader assembles them into a $ScoreList = \{s_{1}, s_{2}, \cdots, s_{c}\}$, and broadcasts it with $VList$ to all members, waiting for the consensus. In this process, each non-faulty member should sign on the $ScoreList$. If successful, the leader sends the agreement to $C_R^r$, together with relevant certification. Consequently, $C_R^r$ updates their reputation by simply adding the listed score.

%\subsubsection*{Phase 3: Incentive on Reputation}
\subsubsection*{Phase 3: Discussion on Incentive}

As shown in~(\ref{eq:score}), as long as a node has higher honest computational resources, it will give a more similar vote with the consensus, thus winning a higher score in each round. Therefore, a node's reputation, or the accumulated score across rounds, is a good reflection of its honest computational resources.

In our design, reputation offers a reference for the leader selection. At the end of any round, $m$ nodes with the highest reputation are selected as leaders of the next round. These nodes are randomly assigned to $m$ committees. As the reputation reflects the honest computational resources one node contributes, such method enhances the performance and throughput of the system.
 
In each round, the profit of nodes is also determined by their reputation. Considering that one's reputation may be negative, we first map the reputation to a positive number using a monotone function $g(\cdot)$, as (\ref{eq:monotone}) shows. Rewards are then distributed proportionally to the mapped value. Such scheme ensures that whoever works more gets more, thus providing enough incentive for nodes to work honestly and as hard as they can.
\begin{equation}
        g(x) = \left\{
            \begin{array}{rcl}
                e^x,       && {x \leq 0;}\\
                1 + \ln(x+1),       && {x > 0.}
            \end{array} 
        \right.
        \label{eq:monotone}
\end{equation}

We also give a punishment mechanism for the misbehavior. For a leader who violates the protocol, %will face a serious penalty. Once a leader is confirmed to commit a fault by the referee committee, his/her 
its reputation will be decreased to the cube root. %Recall that all leaders are those nodes with the highest reputation. We believe that the reputation of each leader is larger than 0, including malicious ones. 
Combining the punishment with \eqref{eq:monotone}, the mapped value, which is closely related to the node's revenue, will reduce to approximately one-third of the original mapped value. Therefore, the higher the reputation a leader has, the stronger the punishment it will suffer when behaving maliciously. 
%In the following section, we will detail the method to find malicious node.

\subsection{Recovery Procedure}\label{sec:sol_rec}

As shown above, the leader of each committee is selected according to reputation. Such scheme enhances the efficiency of the whole protocol by placing computationally-intensive nodes in high-load positions. Nevertheless, it is risky, as malicious nodes with high computational resources may be elected as leaders, by pretending to be honest in early epochs and accumulating high reputation. Therefore, we propose the following recovery procedure to make sure that malicious leaders will not affect the safety and liveness of the protocol. More precisely, we introduce partial set members to always monitoring the behavior of leaders. Once a leader is confirmed to be malicious, it will be evicted by the referee committee collaboratively, and a new leader will be selected.
%To supervise the leaders and enable the system work well even though there exists misbehaving leaders, we propose the following recovery procedure. evict a misbehaving leader and re-select a new one.

\subsubsection{Partial Set and Referee Committee}
%In each committee, $\lambda$ nodes are set to supervise the leader's behavior. These nodes collaboratively is known as the partial set. In each round, the partial set of a committee is selected at random. For security, it is required that at least one node is honest in each partial set. Here, the size $\lambda$ is a parameter. In practice, $\lambda$ can be an appropriate value, like 40.

%Except for the $m$ parallel processing committees, there is a special \textit{Referee Committee} in the system. It's pointed out that nodes in the referee committee are peers. There is neither leader nor partial set. It acts as a mediator to arbitrate between opposing parties. In particular, the referee committee will take action and deal with problems when a certain leader is accused of behaving maliciously. The detailed approach will be given in the following sections.

In each committee, $\lambda$ nodes are set to supervise the leader's behavior. These nodes collaboratively is known as the \textit{partial set}. In each round, the partial set of a committee is selected uniformly at random. For safety, it is required that at least one node is honest in each partial set. In practice, $\lambda$ can be fixed to an appropriate value, like $40$.

Except for the $m$ parallel processing committees, there is a special \textit{referee committee} in the system. It is worth noting that nodes in the referee committee are equal, in the sense that there is neither leader nor partial set in this committee. The referee committee acts as a mediator to arbitrate between opposing parties. Furthermore, it will take action when a certain leader is accused of behaving maliciously. The detailed approach will be given in the following sections.

\figurename~\ref{fig:hie_stru} demonstrates the hierarchical structure of different nodes and committees.

\begin{figure}[!t]
\centering
\includegraphics[width=3.4in]{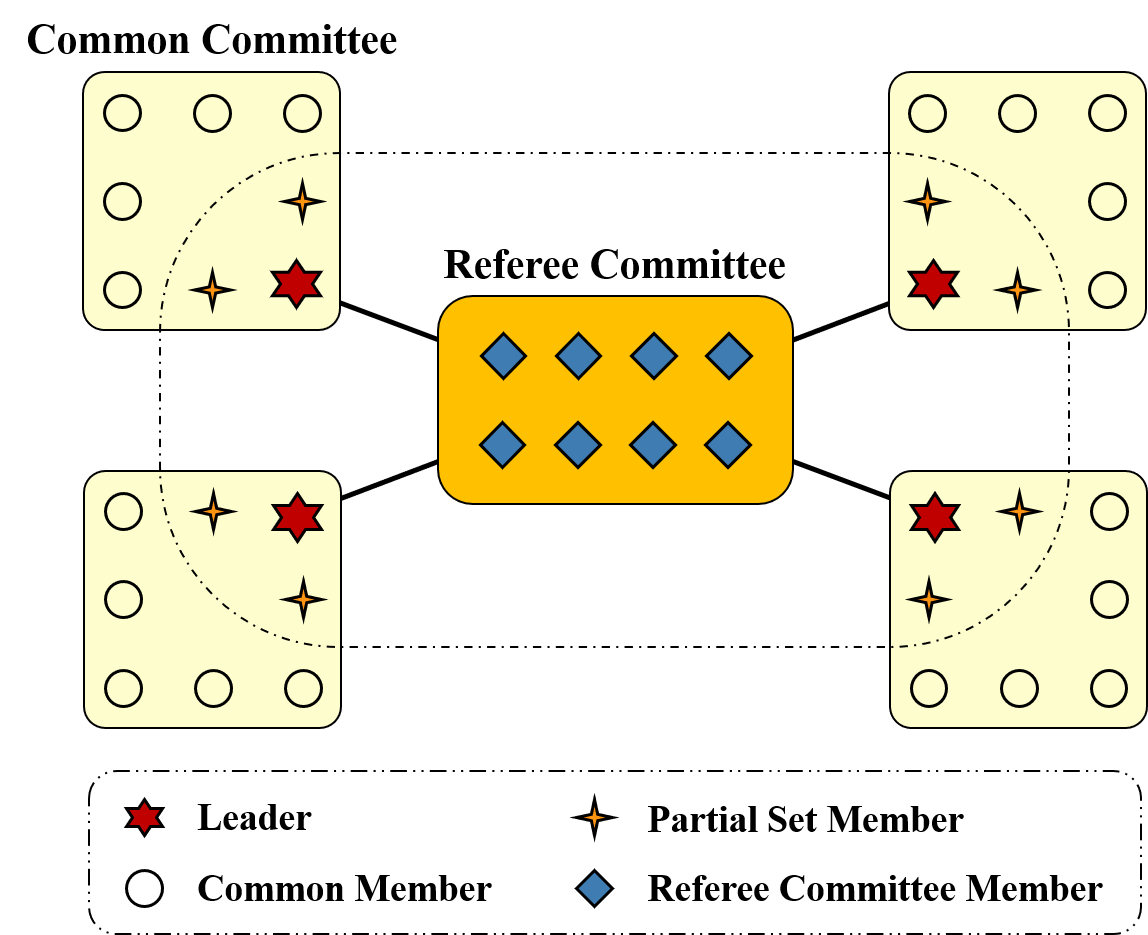}
\caption{Hierarchical structure of the committees. All nodes in a rectangle are fully connected. More specifically, it contains two cases: (a) all nodes in a committee, and (b) all leaders, partial set members and the referee committee.}
\label{fig:hie_stru}
\end{figure}

\subsubsection{Semi-Commitment Exchanging}

In this section, we propose a \textit{semi-commitment exchanging scheme} to prevent a leader from cheating on the member list. Together with the honest partial set members, a malicious leader cannot forge consensus results when communicating with other committees.
%help to constrain leader's behavior and identify malicious leaders

The semi-commitment of a committee is the hash of the member list. In the context of cryptography, a commitment scheme has both hiding property and binding property. Here, we only require the computational-binding property of a commitment scheme. That is where the name "semi-commitment" comes from\footnote{The hiding property of the commitment is not necessary here. For a vicious leader, even if it figures out the members of another committee in the semi-commitment exchanging phase, it can do nothing under our threat model as the adversary cannot get control of a trusty node immediately.}. The semi-commitment exchanging phase runs as follows. 
\begin{enumerate}
    \item To start with, each committee’s leader $l_k$ should unite the member list $S = \{PK_{k,1}, PK_{k,2}, \cdots\ PK_{k,c}\}$ from all key members in the committee, and compute the committee’s semi-commitment via the external hash function $H$: $H_k = H(S).$ Then it broadcasts the semi commitment together with the member list to everyone in the referee committee $C_R$. To prevent any means of cheating, it also delivers both to the partial set $C_{k, par}$.
    \item After all participants in $C_R$ receives semi-commitments from all committees, they run an agreement process inside the committee to check (a) any member in any list is registered; (b) all semi-commitments are valid. They transmit the set of valid semi-commitments to all leaders and partial set members and expel the cheating leaders afterwards.
    \item When a partial set member $c_{k, i}$ gets the semi-commitment $H_k$ from $C_r$, it verifies whether its committee’s semi-commitment corresponds with the member list $S$ it receives from the leader. Once a truthful partial set member notices a mismatch, it reports the circumstance to the referee committee to evict the current leader, to be discussed later.
\end{enumerate} 

We show the phase in Algorithm~\ref{alg:semi-com}. For brevity, the verifying process executed by partial set members as well as all digital signatures are omitted. 
\begin{algorithm}[!t]
    \caption{Semi-commitment Exchange}
    \label{alg:semi-com}
    \renewcommand{\algorithmicwhile}{\textbf{upon}}
    \algsetup{indent=2em}
    \begin{algorithmic}[1]
    %\Procedure{Com\_Exchange}{$r, S, \Psi$}\\
    %\REQUIRE \\
    \ENSURE Each committee get a semi-commitment of any other committee.\\
    \medskip
    \textbf{For leader} $l_k$: 
    \STATE $S \leftarrow \{PK_{k,1}, PK_{k,2}, \cdots\ PK_{k,c}\}$
    \STATE $H_k \leftarrow H(S)$
    \STATE $ComList \leftarrow \mathbf{0}$
    \STATE $ConfList \leftarrow \mathbf{0}$
    \smallskip
    \FOR{$rm \in C_R$}
    \STATE \textit{SEND }$(rm\ |\ \mathsf{SEMI\_COM}, H_k, r, S, k)$
    \ENDFOR
    \smallskip
    \FOR{$pm \in C_{k, par}$}
    \STATE \textit{SEND }$(pm\ |\ \mathsf{SEMI\_COM}, H_k, r, S)$
    \ENDFOR
    \smallskip
    \WHILE{\textit{DELIVER }$(rm\ |\ \mathsf{SEMI\_COM}, H_j, r, j, rm)$}
    \STATE $ConfList[j][H_j] \leftarrow ConfList[j][H_j] + 1$
    \IF{$ConfList[j][H_j] > |C_R| / 2$}
        \STATE $ComList[j] \leftarrow H_j$
    \ENDIF
    \ENDWHILE
    \\
    \medskip
    \textbf{For $rm \in C_R$:}
    \WHILE{\textit{DELIVER }$(l_k\ |\ \mathsf{SEMI\_COM}, H_k, r, S, k)$}
    \STATE $SigList \leftarrow $\textit{ CONSENSUS }$(\mathsf{SEMI\_COM}, r, k, l_k, H_k)$
    \FOR{each leader $l$}
    \STATE \textit{SEND }$(l\ |\ \mathsf{SEMI\_COM}, H_k, r, k, SigList)$
    \ENDFOR
    \ENDWHILE
\end{algorithmic}
\end{algorithm} 

\subsubsection{Leader Re-Selection}

In this section, we introduce the leader re-selection procedure. It is invoked when an honest partial set member notices that its leader is malicious or any participant of $C_R$ notice that some leader is vicious. In the semi-commitment case, the event happens when an honest party (a partial set member or $C_R$) discovers inconsistency between the member list and the semi-commitment.

If a partial set member wants to accuse its leader, it would first broadcast a \textit{witness} to all members in the committee and ask them to vote on the impeachment. Here, a witness is a pair of messages $W = (m_l, m_0)$ where $m_l$ should be sent and signed by the leader. We say a witness is \textit{valid} if and only if the pair derives to a dishonest behavior of the leader. (\textit{e.g.}, $m_l$ be the member list that the leader sends, and $m_0$ be the semi-commitment of the committee where $m_0 \neq H(m_l)$.) If the proposal is approved by more than half of the members, the prosecutor will forward the voting result as well as its witness to everyone in the referee committee.

When any node in $C_R$ receives a witness $W$ and a signature list $Cert$ approving the prosecution from a partial set member $pm$ from committee $C_k$, it starts Algorithm~\ref{alg:re-select} to re-select a committee leader. Afterwards, the new leader needs to make a new semi-commitment of the committee via the semi-commitment exchanging scheme (Algorithm~\ref{alg:semi-com}). When a participant of $C_R$ receives the new semi-commitment, it informs every committee leader the new semi-commitment and leader's address, so that cross-shard transaction handling may start safely.

\begin{algorithm}[!t]
    \caption{Leader Re-selection}
    \label{alg:re-select}
    
    \renewcommand{\algorithmicwhile}{\textbf{upon}}
    \begin{algorithmic}[1]
   
    \ENSURE A malicious leader is evicted and a new leader is selected.\\
    \medskip
    \textbf{For $rm \in C_R$:}
    \STATE $SigList \leftarrow $\textit{ CONSENSUS }$(r, k, pm, W, Cert)$ \\
    \COMMENT{$pm$ is the partial set member who accused the leader, with $W$ the witness and $Cert$ the signatures from committee members.}
    \FOR{each $i \in C_k$}
        \STATE \textit{SEND }$(i\ |\ \mathsf{NEW}, pm, SigList)$
    \ENDFOR
    \end{algorithmic} 
\end{algorithm}  

For a better understanding on how to apply the above designs, readers are recommended to~\cite{ZhangLCCD20} for more details on a complete protocol named CycLedger.

\section{Analysis}\label{sec:analysis}
In this section, we provide a comprehensive discussion on the our design, showing that (1) the committee assignment scheme is secure with overwhelming probability, and (2) the transaction processing procedure satisfies safety and liveness.

\subsection{Security on Committee Assignment}\label{sec:sec_com}

\subsubsection{Partial Set}
%\subsection{Security on Partial Sets}
We say \textit{a partial set is secure} when at least one node in the set is honest. As no more than 1/3 validators are faulty, when the size of the partial set is set to $40$, the probability that a partial set is insecure at most:
    $$ (\frac{1}{3}) ^ {40} < 8 \times 10^{-20}. $$

Associated with union bound, when the number of committees is $20$, the probability that at least one partial set is insecure is no more than $2 \times 10^{-19}$.

\subsubsection{Bootstrapping}
%\subsection{Security on Committee Configuration}
We say \textit{a committee is secure} when more than half of nodes are non-faulty. Recall that committees are formed uniformly except leaders. Let $X$ denote the number of malicious nodes in a committee, and $c$ be the expected committee size. We consider the tail bound of hypergeometric distribution which gives the following result:

\begin{equation}
    \Pr[X \geq \frac{c}{2}] = \sum_{x = \frac{c}{2}}^c \frac{\binom{t}{x}\binom{n - t}{c - x}}{\binom{n}{c}} \leq e^{-D(\frac{1}{2}||f)c},
\end{equation}
where $D(\cdot||\cdot)$ is the Kullback-Leibler divergence. Here $t < \frac{n}{3}$ and $f < \frac{1}{3} + \frac{1}{c}$, thus,

\begin{equation}
    \Pr[X \geq \frac{c}{2}] \leq e^{-\frac{c}{12}}.
    \label{eq:committee_error}
\end{equation}

When the expected committee size is $c = \Theta(\log^2n)$, we derive that the probability that a committee is insecure is less than $n^{\frac{-\log n}{12}}$, which is negligible of $n$.

\figurename~\ref{fig:pro_fail} visualizes \eqref{eq:committee_error}. Namely, it shows the probability of failure calculated using the hypergeometric distribution to uniformly sample a committee when the population of the whole network is 4,000. The amount of malicious nodes is 1,333, exactly less than one-third of the size of the network.

\begin{figure}[t]
    \centering
    \includegraphics[width=3.4in]{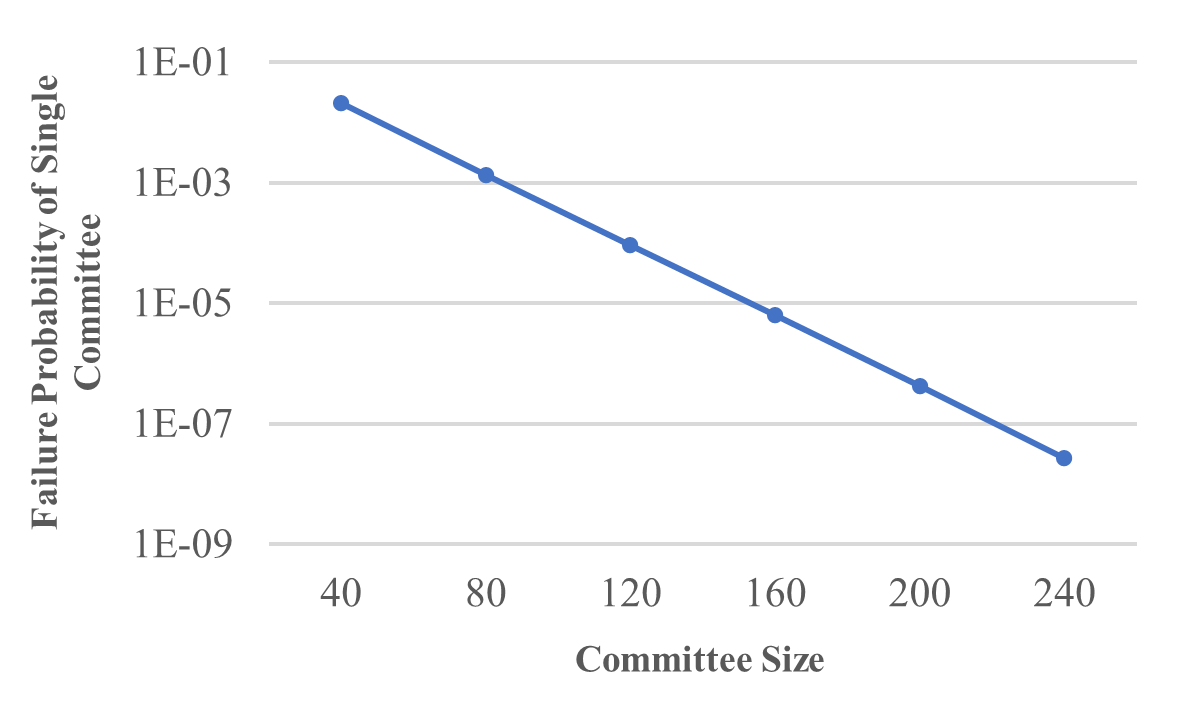}
    \caption{Probability of failure in sampling one committee from a population of 4,000 nodes. The amount of malicious nodes is set to 1,333.}
    \label{fig:pro_fail}
\end{figure}

Particularly, when $c = 240$, the error probability for a single committee is less than $2.8 \times 10^{-8}$. Applying union bound, when $m$ is less than 20, the error probability is no more than $6 \times 10 ^{-7}$.

\subsubsection{Reshuffling}\label{sec:reshuffling}
For the simplicity of expression, we say a round of a sharding protocol is \emph{secure} if all committees in this round are secure. As we have discussed in Section~\ref{sec:issue3}, we use $\alpha$-ECFR scheme to reshuffle the nodes at the beginning of each round. The following theorem shows that after reshuffling with $\alpha$-ECFR scheme, the round is secure with high probability.

\begin{theorem}\label{the:ECFR}
Assume that the adversary controls $f < 1/3$ fraction of nodes, and the corruption requires $d\ (d \geq 1)$ rounds to take effect. Then, there exists a constant $\alpha$, such that under $\alpha$-ECFR scheme, any round is secure with high probability, given that all previous rounds are secure.
\end{theorem}

\begin{proof}
%Recall that under our threat model (see Section~\ref{sec:model}), any corruption attempt by the adversary will go into effect after $d$ rounds. 
To simplify the proof, we suppose that the adversary specifies a set of nodes to corrupt by the end of round $r\ (r > 1)$ (before it knows the committee configuration of round $r + 1$), and nodes in the set become malicious at the start of round $r + d$. To prove the theorem, we show that no matter how the adversary choose the set, nodes will be sufficiently reshuffled after $d$ rounds so that all committees have an honest majority, with high probability.

Further, in this proof, we call a node \textit{black} if it is ever marked in $d$ rounds of reshuffling, and \textit{white} otherwise. 
%The proof proceeds in the following steps:
%\begin{itemize}
%    \item For any $\beta\in (0, 1)$, there is a constant $\alpha\in (0, 1)$, such that after $d$ rounds of $\alpha$-ECFR, with high probability, for all committees, at most $\beta$ fraction of nodes remain white. (Lemma~\ref{lem:ECFR1})
%    \item Conditioning on the committee configuration of round $r$ and that one has the knowledge of black nodes (but not how these nodes move), the permutation of black nodes is uniform at round $r + d$. (Lemma~\ref{lem:ECFR2})
%    \item There is a $\beta\in (0, 1)$, such that whatever the corrupting strategy is at round $r$, all committees have an honest majority at round $r + d$, with high probability. (Lemma~\ref{lem:ECFR3})
%\end{itemize}

\begin{lemma}\label{lem:ECFR1}    
For any constant $\beta\in (0, 1)$, there is a constant $\alpha\in (0, 1)$, such that after $d$ rounds of $\alpha$-ECFR, with high probability, for all committees, at most $\beta$ fraction of nodes remain white.
\end{lemma}
\begin{proof}
Consider any committee $C$ with size $c = O(\log^2 n)$. Let $X_1, \cdots, X_i,\cdots, X_d$ denote the number of newly-marked nodes in $C$ in round $r + 1, \cdots, r + i, \cdots, r + d$. Let $\beta := (1 - \alpha^2)^d$, and $\gamma_i := \left(1 - (1 - \alpha^2)^i\right), 1 \leq i \leq d$. (Therefore, $\gamma_d = 1 - \beta$). We define the following $d$ events:
\[\mathcal{E}_i := \{X_1 + \cdots + X_i > \gamma_i\cdot c\},\quad 1 \leq i \leq d.\]
%\begin{align*}
%    \mathcal{E}_1 &:= (X_1 > \gamma_1\cdot c), \quad \cdots\\
%    \mathcal{E}_i &:= (X_1 + \cdots + X_i > \gamma_i\cdot c), \quad \cdots\\
%    \mathcal{E}_d &:= (X_1 + \cdots + X_d > \gamma_d\cdot c).
%\end{align*}

First, notice that $E[X_1] = \alpha\cdot c$. By Chernoff-Hoeffding bound, we have
\[
    \Pr[\neg \mathcal{E}_1] = \Pr[X_1 \leq \alpha^2\cdot c] \leq \exp\left\{-\frac{1}{2}\cdot (1 - \alpha)^2\alpha\cdot c\right\}.
\]

For any $1 \leq i < d$ (when $d \geq 2$), We now give an upper bound on $\Pr[\neg \mathcal{E}_{i + 1}]$.
\begin{align*}
    &\quad\ \Pr[\neg \mathcal{E}_{i + 1}] \\
    &= \Pr[\neg \mathcal{E}_{i + 1} \wedge \neg \mathcal{E}_{i}] + \Pr[\neg \mathcal{E}_{i + 1} \wedge \mathcal{E}_{i}] \\
    &\leq \Pr[\neg \mathcal{E}_i] + \sum_{t > \gamma_i\cdot c}\Pr[\neg \mathcal{E}_{i + 1} \wedge (X_1 + \cdots + X_i = t)] \\
    &\leq \Pr[\neg \mathcal{E}_i] + \Pr[\mathcal{E}_i]\cdot \max_{t > \gamma_i\cdot c}\Pr[\neg \mathcal{E}_{i + 1} | X_1 + \cdots + X_i = t] \\
    &\leq \Pr[\neg \mathcal{E}_i] + \max_{t > \gamma_i\cdot c}\Pr[X_{i + 1} \leq \gamma_{i + 1}\cdot c - t | X_1 + \cdots + X_i = t].
\end{align*}

An important property is that $\gamma_{i + 1} - \alpha \leq (1 - \alpha)\gamma_i$, which implies that $\gamma_{i + 1}\cdot c - t \leq \alpha\cdot (c - t)$ for any $t > \gamma_i\cdot c$. As $E[X_{i + 1}] = \alpha\cdot (c - X_1 - \cdots - X_i)$, we can apply Chernoff-Hoeffding bound again:
\begin{align*}
    &\quad\ \Pr[\neg \mathcal{E}_{i + 1}] - \Pr[\neg \mathcal{E}_i]\\
    &\leq \max_{t > \gamma_i\cdot c}\Pr[X_{i + 1} \leq \gamma_{i + 1}\cdot c - t | X_1 + \cdots + X_i = t] \\
    &\leq \max_{t > \gamma_i\cdot c}\exp\left\{-\frac{1}{2}\cdot \left(1 - \frac{\gamma_{i + 1}\cdot c - t}{\alpha\cdot (c - t)}\right)^2\cdot \alpha\cdot (c - t)\right\} \\
    &= \max_{t > \gamma_i\cdot c}\exp\left\{-\frac{1}{2}\cdot \frac{\left((\alpha - \gamma_{i + 1})\cdot c + (1 - \alpha)\cdot t\right)^2}{\alpha\cdot (c - t)}\right\} \\
    &< \exp\left\{-\frac{1}{2}\cdot \frac{\left((\alpha - \gamma_{i + 1})\cdot c + (1 - \alpha)\cdot \gamma_{i}\right)^2}{\alpha\cdot (c - \gamma_{i})}\cdot c\right\} \\
    &= \exp\left\{-\frac{1}{2}\cdot (1 - \alpha)^2\alpha(1-\alpha^2)^{i}\cdot c\right\}.
\end{align*}

Therefore, as $c = \Theta(\log^2 n)$, we have
\begin{align*}
    \Pr[\neg \mathcal{E}_d] &= \Pr[\neg \mathcal{E}_1] + \sum_{i = 1}^{d - 1}(\Pr[\neg \mathcal{E}_{i + 1}] - \Pr[\neg \mathcal{E}_{i}]) \\
    &\leq \sum_{i = 0}^{d - 1} \exp\left\{-\frac{1}{2}\cdot (1 - \alpha)^2\alpha(1-\alpha^2)^{i}\cdot c\right\} \\
    &\leq d\cdot \exp\left\{-\frac{1}{2}\cdot (1 - \alpha)^2\alpha(1-\alpha^2)^{d - 1}\cdot c\right\} \\
    &= d\cdot \exp\left\{-\frac{1}{2}\cdot \frac{(1 - \alpha)\alpha}{1 + \alpha}\beta\cdot c\right\} = \Theta(n^{-\log n}).
\end{align*}

Note that $\mathcal{E}_d$ is the event that at most $\beta$ fraction of nodes in committee $C$ remain white. By a union bound on all $m = \Theta(n/\log^2 n)$ committees, the lemma is proved.
\end{proof}

\begin{lemma}\label{lem:ECFR2}
Conditioning on the committee configuration of round $r$ and the number of black nodes, the identity of these black nodes and their belongings at round $r + d$ is uniform.
\end{lemma}
\begin{proof}
The lemma holds obviously according to the definition of ECFR scheme.
\end{proof}

Now we come back to the main theorem. Let $Y$ be the number of black nodes, and $Z$ be the number of black nodes that are to be corrupted at round $r + d$. By Lemma~\ref{lem:ECFR1}, with high probability, $Y \geq (1 - \beta)\cdot n$.

We simply disregard the negligible failure probability. Due to Lemma~\ref{lem:ECFR2}, $Z$ follows the hypergeomatric distribution $\mathcal{H}(f\cdot n, n, Y)$. Therefore, we have
\begin{align*}
    &\quad \ \Pr[Z \geq f\cdot (1 + \beta)\cdot Y] \\
    &\leq \max_{y\geq (1 - \beta)\cdot n}\Pr[Z\geq f\cdot (1 + \beta)\cdot y | Y = y] \\
    &\leq \exp\left\{-D\left(f\cdot (1 + \beta) \| f\right)\cdot (1 - \beta)\cdot n\right\},
\end{align*}
according to the tail bound of hypergeometric distribution.

Now suppose $Z \leq f\cdot (1 + \beta)\cdot Y$, which happens with high probability according to the previous inequality. For a committee $C$, let $Y_C$ be the number of black nodes in $C$ at round $r + d$, and $Z_C$ be the number of black nodes in $C$ that are malicious at round $r + d$. Again, $Z_C\sim \mathcal{H}(Z, Y, Y_C)$, which leads to
\begin{align*}
    &\quad \ \Pr[Z_C \geq f\cdot (1 + \beta)^2\cdot Y_C] \\ &\leq \max_{y_C\geq (1 - \beta)\cdot c}\Pr[Z_C\geq f\cdot (1 + \beta)^2\cdot y_C | Y_C = y_C] \\
    &\leq \exp\left\{-D\left(f\cdot (1 + \beta)^2\|f\cdot (1 + \beta)\right)\cdot (1 - \beta)\cdot c\right\},
\end{align*}

which is negligible in $n$. Subsequently, we assume $Z_C \leq f\cdot (1 + \beta)^2\cdot Y_C$. Let $M_C$ be the number of malicious nodes in committee $C$ at round $r + d$. Above all, with probability $1 - O(n^{-\log n})$, we have
\begin{align*}
    M_C &\leq \beta \cdot c + (1 - \beta)\cdot f\cdot (1 + \beta)^2 \cdot c \\
    &= \left(\beta + f\cdot (1 + \beta)^2(1 - \beta)\right)\cdot c.
\end{align*}
    
When $f < 1/3$, with appropriate small $\beta$ (\textit{e.g.}, $\beta = 1/8$), we have $\beta + f\cdot (1 + \beta)^2(1 - \beta) < 1/2$. Applying an union bound on all $m = \Theta(n/\log^2 n)$ committees, we obtain the theorem.
\end{proof}
%The proof is deferred to Appendix~\ref{app:ECFR}. Therefore, within polynomial rounds of $n$, all committees in each round is secure with high probability.

\subsection{Safety and Liveness on Transaction Processing}
Safety and liveness are two classes of essential properties in sharding blockchain system, with the following implication respectively:
\begin{itemize}
    \item \textit{Safety.} Each committee will never propose a block with invalid transactions.
    \item \textit{Liveness.} By the end of each round, each committee will propose a non-empty valid block.
\end{itemize}

In this section, we show that our design satisfies both two properties.
%w\subsection{Security on Randomness}

%In CycLedger, we apply the SCRAPE scheme~\cite{CascudoD17} within $C_R$ to distributedly generate a random string. SCRAPE guarantees that as long as the majority of nodes in $C_R$ are honest, the output random string is pseudorandom and unpredictable. At the same time, no leader is required in the execution of SCRAPE. This feature suits the construction of the referee committee well as $C_R$ is the only committee without a leader. As one can see in the following analysis, each committee, including $C_R$, has more than half of non-faulty nodes with high probability, hence, we assert that the randomness produced by SCRAPE with $C_R$ is reliable.

%\subsection{Security on Semi-Commitments}

%To start with, we show that 
%\begin{lemma}
%When $H$ is modeled as a collision-resistant hash function (CRHF), $SEMI\_COM_k^r$ satisfies the computational binding property, \textit{i.e.}, 
%\label{claim:release_com}
%\end{lemma}

\begin{claim}\label{cla:mem}
A malicious leader cannot deceive a trustful leader by forging a member list of its committee as long as the referee committee has an honest majority.
\end{claim}

\begin{proof}
The process of semi-commitment exchanging is under the supervision of its partial set. Note that with high probability the partial set has at least one honest node. Therefore, each leader cannot lie and the semi-commitment exactly corresponds to the true member list.

Owing to the collision-resistance property of the hash function, the semi-commitment of a committee satisfies the computational binding property. After the semi-commitment is released, only with negligible probability, a probabilistic polynomial-time malicious leader can forge a false member list which corresponds to the same semi-commitment. Therefore, the leader cannot provide non-accepted results by falsifying its committee signature.
%There are only two opportunities when a betrayer can lie to a loyalist on the member list. However, because all members in $C_R$, as well as the partial set members of the committee, see the exact member list, any false hash on the list will be perceived. Thus the liar will be detected. The other chance for the malicious leader is when the semi-commitment is revealed. However, according to the Lemma~\ref{claim:release_com}, misleading behavior will not take effect in this phase.
\end{proof}

\begin{claim}\label{cla:com}
A malicious leader is always detected and thus evicted via the leader re-selection procedure, as long as the referee committee has an honest majority.
\end{claim}

\begin{proof}
According to the discussion in Section~\ref{sec:sec_com}, with high probability there is at least one honest node in the partial set and the referee committee has an honest majority. Therefore, as a leader's action is always monitored by the partial set during the execution of the system, any irregular behavior from the leader will be detected and a witness will be inevitably grasped by the non-faulty partial set member. At the same time, as the evidence is signed by the leader itself, a malicious leader can never deny the charges.
\end{proof}

%Now we claim that the given procedure is both complete and sound:

\begin{claim}\label{cla:so}
A trustful leader will never be framed up by a faulty partial set member, as long as $C_R$ has an honest majority.
\end{claim}

\begin{proof}
We mention that a witness is valid if and only if the first part of it is a message signed by the leader. For the security of the digital signature scheme, a faulty partial set member cannot counterfeit a shred of evidence which must be a leader's signed message. Therefore, a trustful leader will never be unjustly accused.
\end{proof}

%As proved above, the probability that more than half members in $C_R$ are faulty is negligible. Thus, we claim that our recovering procedure remains complete and sound with high probability.
As proved above, any leader can never behave badly, such as tampering with cross-shard transactions, proposing blocks with invalid transactions and so on. Otherwise, it will be evicted until an honest node becomes the leader. As a consequence, a non-empty valid block will be proposed by the end of each round. Thus, we claim that our design has both the safety property and liveness property.

\section{Simulation}\label{sec:simulation}
% 假设：每个节点的算力不同
% 实验主体：两个设计对比，①leader随机选，②leader根据reputation选
% 衡量的指标：两种设计单位时间内处理的交易量（tps）
% Todo: 设计好假设和setting
% 比如，一个committee处理的tps与leader的算力正相关
% 1. 节点的诚实算力服从P分布
% 2. committee处理的交易数与leader的算力
% 3. 

% 我们的模拟逻辑：
% 1. 目的：验证reputation是否能帮我们选出诚实算力较高的leader，从而提高处理交易的效率
% 2. 算力建模：每个节点的诚实算力量

% 这里随手记一些逻辑：
% 1. 我们的机制保证最终作恶者不能占据每一个committee的较多数
% 2. 从而每笔交易结果一定正确
% 3. 任何一笔交易，只要其被任何一个节点的诚实算力处理过，则这一交易最终一定被处理，并且结果与任何一节点诚实算力处理后的结果相同
% 4. 在协议中，因为inter-committee的交易会被拆分成intra-committee的交易，所以在模拟时只考虑后者的存在

In this section, we conduct several simulations to evaluate the performance of our reputation mechanism.
\subsection{Setup}
%As discussed in Section~\ref{sec:analysis}, the theoretical analysis proves that our design is secure. More specifically, with overwhelming probability, each committee has an honest majority at any time and only valid transactions will be accepted. For completeness, this section attaches more importance to the performance evaluation. In practice, there is no well-developed and proven sharding blockchain system yet. And it is infeasible and kind of risky to apply a novel design to real systems. As an alternative, we adopt the simulation method to evaluate the performance of our reputation mechanism. 

As discussed in Section~\ref{sec:analysis}, the theoretical analysis proves that our design is secure. More specifically, with overwhelming probability, each committee has an honest majority at any time and only valid transactions will be accepted. For completeness, this section attaches more importance to the performance evaluation. In practice, there is no well-developed and proven sharding blockchain system yet, and it is infeasible and kind of risky to apply a novel design to real systems. As an alternative, we adopt the simulation method to evaluate the performance of our reputation mechanism.

There are several basic assumptions in our simulation. 
\begin{itemize}
    \item As mentioned previously, sharding protocols work in rounds. For convenience, we discretize the time into time slots with a fixed span.
    \item The resource cost of a transaction is assumed to be linear with the number of its inputs, which is realistic in real-life transaction-processing.
    \item Leaders suffer from extra resource consumption on the decision making, reputation updating and other additional tasks, and only a $p_l$ fraction of computation resources is left to process transactions. Meanwhile, common members are also burdened with communication in reaching consensus, therefore with a $p_m$ fraction of computation resources left to process transactions. We suppose that $p_l < p_m$.
\end{itemize}

\begin{figure}[t]
    \centering
    \includegraphics[width=3.4in]{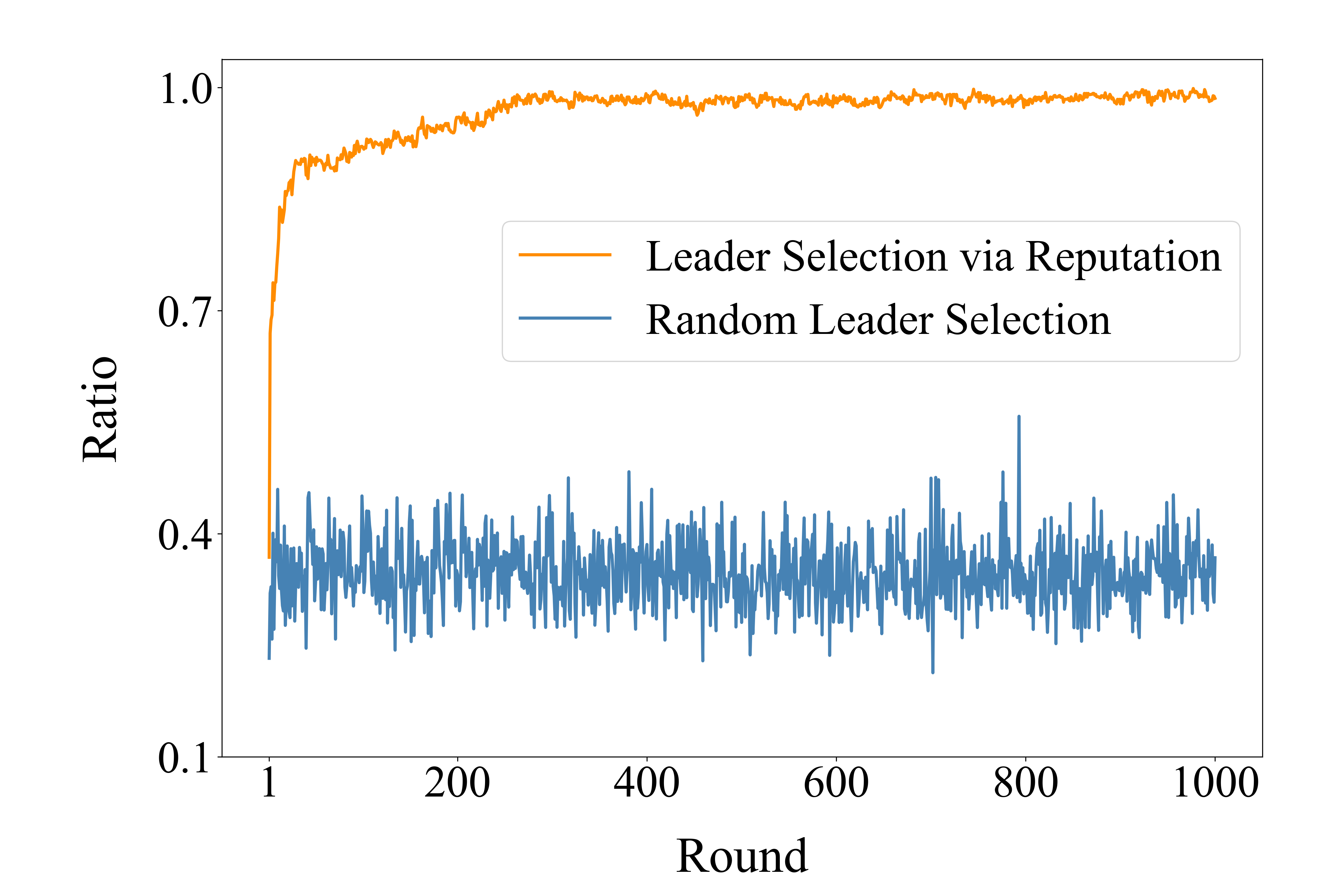}
    \caption{The ratio of average leader computer resources to the maximal possible average leader computer resources against round number under two types of leader selection schemes. The total node number is 2,000 and the number of committees is 20. The orange line shows the case of leader selection via reputation, while the blue line shows the case of random leader selection.}
    \label{fig:aver_leader_com}
\end{figure}

Without loss of generality, we mainly consider intra-shard transactions. Essentially speaking, the processing of cross-shard transactions also depends on the intra-committee consensus\footnote{The core step to process cross-shard transactions is the intra-committee voting on the status of relevant inputs.}, And intra-shard transactions and cross-shard transactions cause no difference on the reputation-updating process, which happens after voting on the transaction/input list.

In order to emphasize our focus, we neglect the possible misbehaviors and assume that the honest computational resources of nodes in the system follow a beta distribution, which suits the reality well. The resource cost of a transaction is assumed to be positively associated with the number of its inputs. When processing transactions, leaders suffer from extra resource consumption on the decision making, reputation updating and other additional tasks.

In simulation, at the beginning of each round, 2,000 nodes are uniformly split into 20 committees, each with size 100. In order to imitate the realistic condition better, we crawl 208,936 transactions from Bitcoin network, and ignore those with more than 12 inputs (which barely happen in today's network, see\figurename~\ref{fig:inputs_num_cdf}). Resembling reality, these transactions are sent to the system at a certain frequency. For each transaction $\mathsf{tx}$, it is assigned to the committee, the ID of which is $\mathsf{tx\_hash} \bmod 20$.

For the specific parameters, we set the remaining fraction of computation resources to process transactions for leaders and common members to be $p_l = 0.7$ and $p_m = 0.9$ correspondingly. Further, for the two parameters in bonus part of the reputation updating process (see~(\ref{eq:bonus})), we set $\sigma = 0.1$ and $\omega = 0.5$.

In order to observe the effect of the reputation mechanism, we evaluate it using the following metrics.
\begin{itemize}
    \item \textbf{Leaders' computation resource: } the average computation resources of all leaders in one round.
    \item \textbf{Transactions processed per round: } number of the processed transactions in one round.
\end{itemize}
We mainly compare the outcomes of leader selection scheme via reputation and random leader selection scheme, which is widely adopted by other sharding blockchain solutions~\cite{Kokoris-KogiasJ18}.

\subsection{Results}

We run the simulation for 1,000 rounds for both cases: leader selection scheme via reputation and random leader selection scheme. \figurename~\ref{fig:aver_leader_com} shows the ratio of the average computation resources of 20 leaders to the average computation resources of 20 most computationally-intensive nodes against round number under both schemes. For the random leader selection scheme, the ratio fluctuates with an expected value of approximately 0.37. Meanwhile, for our leader selection scheme via reputation, the ratio continuously rises and reaches 1 in 1,000 rounds. There are two reasons for the small oscillation on the ratio for our scheme: (1) leaders suffer more loss on computation resources than common members, and (2) leaders gain less reputation than those common members with high computation resources. As a result, those nodes whose computation resources are a bit lower than maximal also have the chance to process all transactions provided by the leader, therefore gaining higher reputation, and becoming the leader.

\begin{figure}[t]
    \centering
    \includegraphics[width=3.4in]{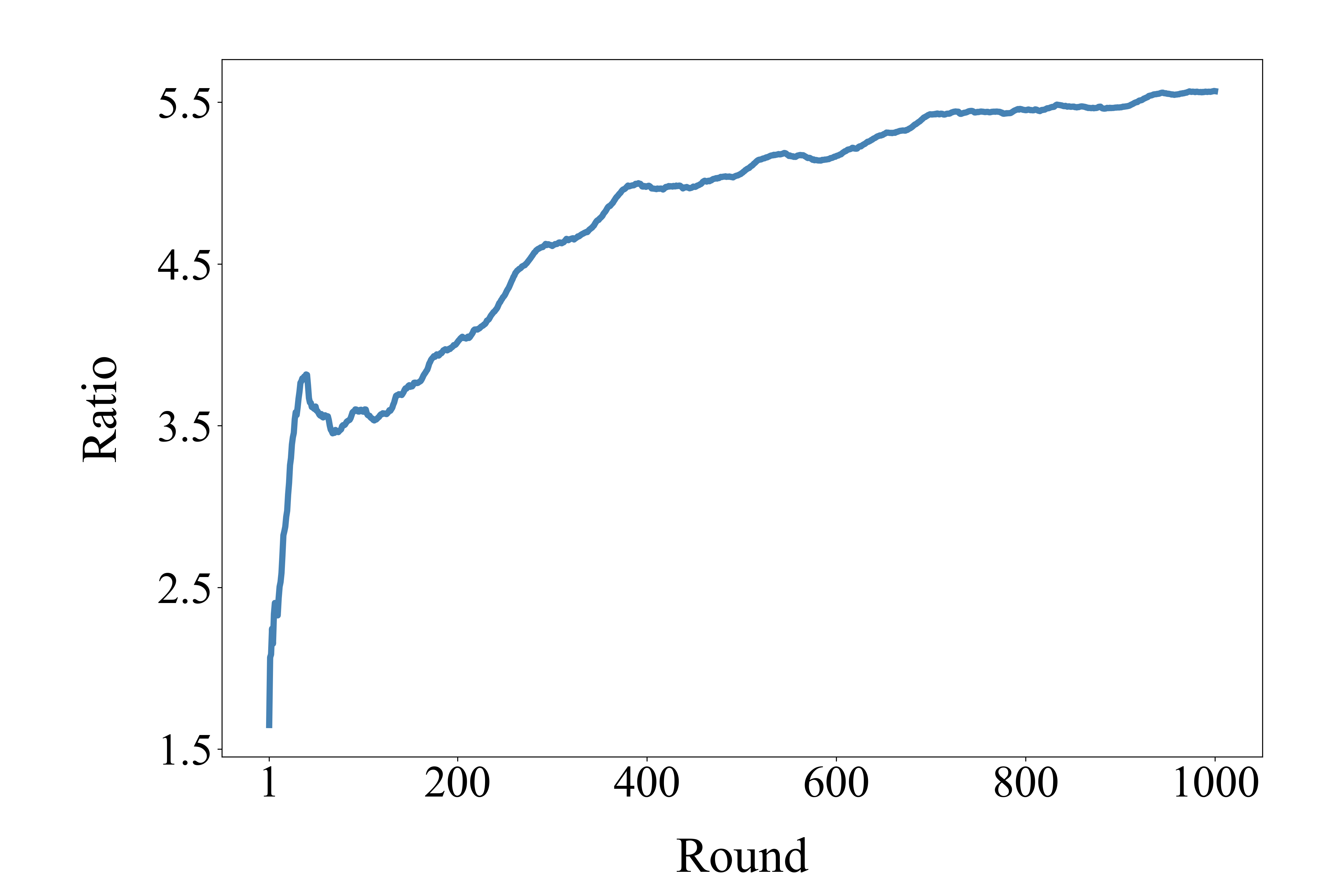}
    \caption{The ratio of accumulated number of transactions processed under two types of leader selection schemes against round number. The total node number is 2,000 and the number of committees is 20. The numerator is the corresponding number for the \emph{leader selection scheme via reputation}, and the denominator is the corresponding number for the \emph{random leader selection scheme}.}
    \label{fig:tran_num_ratio}
\end{figure}

\figurename~\ref{fig:tran_num_ratio} shows the ratio of accumulated number of transactions processed under two types of leader selection schemes. It turns out that the ratio continuously rises in 1,000 rounds. By the end of round 1,000, the number of processed transactions under leader selection scheme with reputation is approximately $5.5 \times$ the corresponding number under the random leader selection scheme. Such result is a direct corollary of the previous result. As leaders are with higher computation resources, they can include more transactions in the consensus in each round, and therefore processing more transactions.

Taken together, our leader selection scheme via reputation successfully picks out those nodes with higher computation resources and makes them leaders, therefore dramatically improves the transaction-processing speed.

\section{Related Work}\label{sec:related}

Elastico~\cite{LuuNZBGS16} is the first sharding-based protocol for public blockchains which can tolerate up to a fraction of 1/4 of malicious parties. Unfortunately, it has a very weak safety guarantee as the randomness in each epoch of the protocol can be biased by the adversary. Meanwhile, Elastico's small committees (only about 100 nodes in a committee) cause a high probability to fail under a 1/4 adversary, and cannot be released in a Proof-of-Work (PoW) system~\cite{GervaisKCC14}. Specifically, when there are 16 shards, the failure probability is 97\% over only 6 epochs~\cite{Kokoris-KogiasJ18}. 
OmniLedger~\cite{Kokoris-KogiasJ18} also allows the adversary to take control of at most 25\% of the validators as well as assuming the adversary to be mildly-adaptive, nevertheless, it depends on the assumption that there is a never-absent trusty client to schedule the leaders' interaction when handling cross-shard transactions. 
RapidChain~\cite{ZamaniM018} enhances the efficiency of sharding-based blockchain protocols on a large scale, but the protocol guarantees high efficiency only when leaders of each committee are honest, an unrealistic assumption in practice. Concretely, in expectation, there is a proportion of 1/3 leaders that are malicious in a round. Under this condition, cross-shard transactions may hardly be included in a block. Furthermore, the protocol does not have an explicit incentive for nodes to participate in. At the same time, all the above postulate a good connection between any pair of truthful nodes, which causes a huge burden in creating connection channels. 

%Different from other sharding protocol, Ostraka~\cite{ManuskinME20} shards the nodes inside, embracing the fact that there are several super nodes with enough computing resources. It provides a parallel algorithm for these super nodes to validate transactions concurrently among multi processors. It is proved indistinguished with original topological-order transactions verification algorithm in Bitcoin. However, the sharding algorithm only applies to nodes with abundant computing power.

Optchain~\cite{NguyenNDT19} provides a transaction assignment algorithm to reduce the number of cross-shard transactions, while keeping load balance among shards. It mainly uses a PageRank-like algorithm to assign linked transactions to the same shard. The disadvantage is that this algorithm is client-driven, requiring users to query information and do much computation, but without any incentive.
\cite{LiuLYLYW20} focuses on the processing of cross-shard transactions, which is of vital importance to the system efficiency. By constructing multiple inputs into a Merkle tree structure, the number of BFT calls is largely reduced. %However, such design suffers from a poor liveness once any leader is dishonest.

All above work fail to maintain liveness concerning malicious committee leaders or transaction coordinators. ~\cite{DangDLCLO19} solves this problem by involving a reference committee, with an honest majority, and uses two-phase commit (2PC) and two-phase locking (2PL) protocols to coordinate all cross-shard transactions. However, such solution brings a heavy overhead on the reference committee, which is a major shortcoming. 
SSChain~\cite{ChenW19} introduces a root chain to process cross-shard transactions. It solves the liveness problem and ensure the system safety. However, the root chain is easily corrupted. Furthermore, both above two solutions cause a large reduce on the parallelism brought by sharding, as the reference committee/root chain needs to process cross-shard transactions sequentially.

\section{Conclusion}\label{sec:conclusion}

We identify five basic issues in sharding blockchain. Specifically, we analyze the challenges involved and present our suggested solutions. In order to overcome the performace bottlenecks caused by cross-shard transactions, we introduce the concept of reputation and propose a reputation mechanism. By scoring each node according to its historical behavior, the term of reputation helps to locate those nodes with more honest computational resources. By assigning them to high-workload positions, the reputation mechanism enhances the system's capability to process transactions. In addition, we introduce a semi-commitment scheme and a recovery procedure. They together enable users to detect and evict malicious leaders, thus trading safely in the system. Theoretical analysis and simulation results confirm that our design can significantly improve the system performance, without sacrificing any safety and liveness.

\ifCLASSOPTIONcaptionsoff
  \newpage
\fi



%\bibliographystyle{IEEEtran}
\bibliographystyle{IEEEtran}
\bibliography{main}

\begin{IEEEbiography}[{\includegraphics[width=1in,height=1.25in,clip,keepaspectratio]{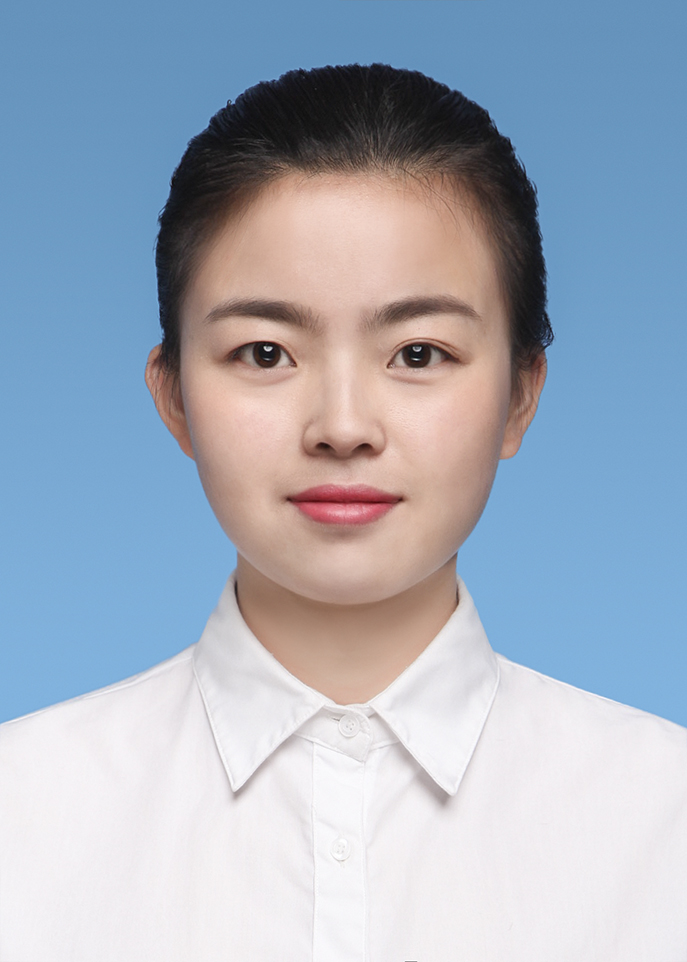}}]{Mengqian Zhang} received the B.S. degree in computer science from Ocean University of China
in 2018. She is currently pursuing the doctor’s
degree at the Department of Computer Science
and Engineering, Shanghai Jiao Tong University,
China. Her current research interests include Blockchain, Algorithmic Game Theory and Mechanism Design.
\end{IEEEbiography}
\begin{IEEEbiography}[{\includegraphics[width=1in,height=1.25in,clip,keepaspectratio]{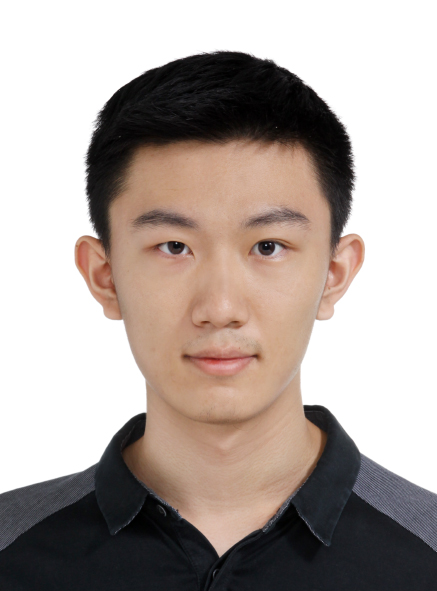}}]{Jichen Li} received the B.S. degree in computer science at School of EECS, Peking University in 2020. He is currently pursuing the doctor’s degree at Center on Frontiers of Computing Studies, Peking University. His current research interests focus on Blockchain and Mechanism Design.
\end{IEEEbiography}
\begin{IEEEbiography}[{\includegraphics[width=1in,height=1.25in,clip,keepaspectratio]{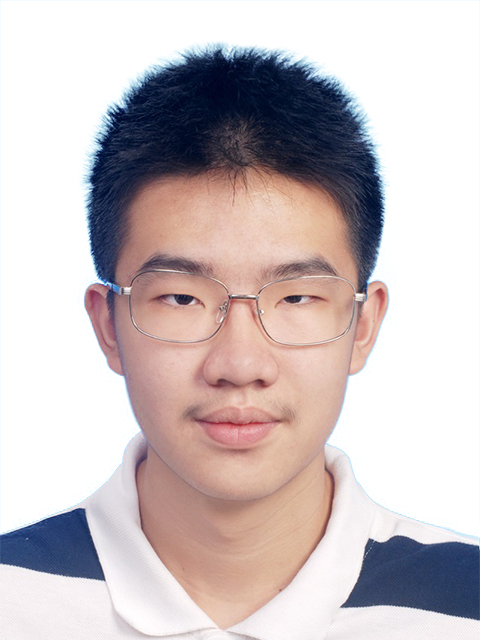}}]{Zhaohua Chen} is now a fourth-year undergraduate student at School of EECS, Peking University. His research interest lies in various subjects of theoretical computer science.
\end{IEEEbiography}
\begin{IEEEbiography}[{\includegraphics[width=1in,height=1.25in,clip,keepaspectratio]{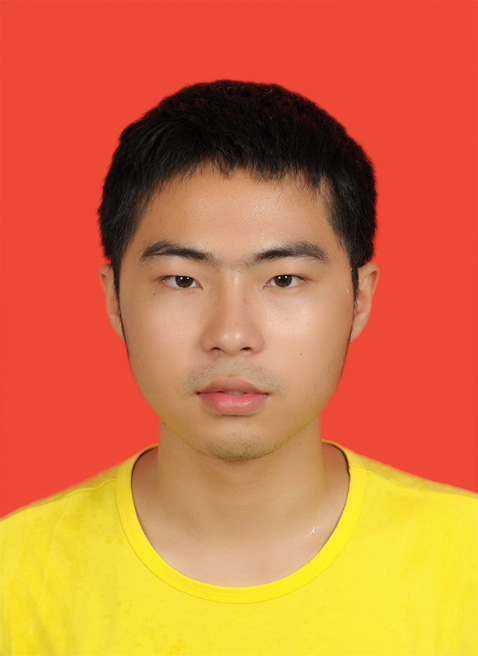}}]{Hongyin Chen} received the B.S. degree at School of EECS, Peking University in 2020. He is currently pursuing the doctor’s
degree at Center on Frontiers of Computing Studies, Peking University. His current research interests include Blockchain and Mechanism Design.
\end{IEEEbiography}
\begin{IEEEbiography}[{\includegraphics[width=1in,height=1.25in,clip,keepaspectratio]{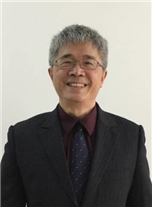}}]{Xiaotie Deng} got his BSc from Tsinghua University, MSc from Chinese Academy of Sciences, and PhD from Stanford University in 1989. He is currently a chair professor at Peking University. He taught in the past at Shanghai Jiaotong University, University of Liverpool, City University of Hong Kong, and York University. Before that, he was an NSERC international fellow at Simon Fraser University. Deng’s current research focuses on algorithmic game theory, with applications to Internet Economics and Finance. His works cover online algorithms, parallel algorithms, and combinatorial optimization. He is an ACM fellow for his contribution to the interface of algorithms and game theory, and an IEEE Fellow for his contributions to computing in partial information and interactive environments.
\end{IEEEbiography}


%\newpage



%\vfill

% is flush with the other column.\bigcup

\end{document}